\documentclass[10pt]{article}

\usepackage[T1]{fontenc}
\usepackage[utf8]{inputenc}
\usepackage{lmodern}
\usepackage[margin=1in]{geometry}
\usepackage{amsmath,amssymb,amsthm}
\usepackage{booktabs}
\usepackage{siunitx}
\usepackage{xcolor}
\usepackage{graphicx}
\usepackage{microtype}
\usepackage{caption}
\usepackage{bytefield}
\usepackage{placeins}
\usepackage{float}
\usepackage{tikz}
\usetikzlibrary{positioning,calc,fit,backgrounds,arrows.meta,shadows.blur,shapes.geometric}
\usepackage[colorlinks=true,allcolors=linkblue]{hyperref}

\definecolor{ink}{HTML}{1B1F24}
\definecolor{linkblue}{HTML}{2A5DB0}
\definecolor{obsfill}{HTML}{EAF0FA}
\definecolor{obsline}{HTML}{3A5FB0}
\definecolor{redfill}{HTML}{F2ECF8}
\definecolor{redline}{HTML}{6A4C9C}
\definecolor{clmfill}{HTML}{FBEFE6}
\definecolor{clmline}{HTML}{C8743A}
\definecolor{docfill}{HTML}{E9F4EE}
\definecolor{docline}{HTML}{2F7D4F}
\definecolor{cS0}{HTML}{3A5FB0}
\definecolor{cS1}{HTML}{2F7D4F}
\definecolor{cS2}{HTML}{C8743A}
\definecolor{soft}{HTML}{6B7480}

\newtheorem{definition}{Definition}
\newtheorem{lemma}{Lemma}
\newtheorem{proposition}{Proposition}
\newtheorem{theorem}{Theorem}

\newcommand{\hx}[1]{{\ttfamily\scriptsize\textcolor{soft}{#1}}}
\newcommand{\class}[1]{\textsf{\small #1}}
\newcommand{\code}[1]{\texttt{\small #1}}
\DeclareSIUnit\flop{FLOP}
\DeclareSIUnit\byte{B}

\newcommand{\egDevice}{RTX PRO 6000 Blackwell Max-Q}

\newcommand{\egCardMemGiB}{95}
\newcommand{\egSMCount}{188}

\newcommand{\egComputeCap}{12.0}
\newcommand{\egDriver}{580.126.09}
\newcommand{\egCuda}{13.0}
\newcommand{\egTorch}{2.12.1}
\newcommand{\egRootHash}{1303d007503848af}

\newcommand{\egNumObservations}{240}
\newcommand{\egMaxRepeats}{24}

\newcommand{\egGemmBfOneSixMedian}{248.4}

\newcommand{\egGemmFpEightMedian}{484.2}

\newcommand{\egScatterNd}{\num{3.51e-05}}
\newcommand{\egIndexNd}{\num{7.55e-05}}
\newcommand{\egK}{8}
\newcommand{\egVMargin}{3}
\newcommand{\egVtau}{\num{6.06e-04}}
\newcommand{\egPmiss}{\num{3.91e-03}}

\newcommand{\egFloorSizes}{2048, 4096, 8192, and 16384}
\newcommand{\egFloorSeqFpSixteen}{\num{1.95e-04}, \num{2.02e-04}, \num{1.93e-04}, and \num{2.02e-04}}
\newcommand{\egFloorSeqBfSixteen}{\num{1.51e-03}, \num{1.78e-03}, \num{1.53e-03}, and \num{1.69e-03}}
\newcommand{\egFloorSeqTfThreeTwo}{\num{2.71e-04}, \num{3.56e-04}, \num{2.86e-04}, and \num{2.89e-04}}
\newcommand{\egFloorSeqFpEight}{\num{1.63e-03}, \num{1.78e-03}, \num{1.61e-03}, and \num{1.71e-03}}

\newcommand{\egFlipResid}{\num{2.01e-03}}
\newcommand{\egCrossDelta}{0}
\newcommand{\egTputDiffFpSixteen}{2.56}
\newcommand{\egTputDiffFpEight}{2.21}
\newcommand{\egFillMem}{78}

\newcommand{\egFillTput}{1515}
\newcommand{\egWitnessDeltaCorrupt}{\num{1.59e-01}}
\newcommand{\egAttFloor}{\num{2.00e-03}}
\newcommand{\egAttFault}{\num{2.22e-01}}
\newcommand{\egAttRowsum}{\num{2.38e-07}}
\newcommand{\egScatterFloor}{\num{1.03e-07}}
\newcommand{\egScatterFault}{\num{7.95e-04}}
\newcommand{\egPortableResid}{\num{1.97e-04}}
\newcommand{\egPortableTau}{\num{7.50e-04}}
\newcommand{\egValDevice}{RTX 5090}

\newcommand{\egValCardMemGiB}{31}
\newcommand{\egValSMCount}{170}
\newcommand{\egValComputeCap}{12.0}
\newcommand{\egValDriver}{595.58.03}
\newcommand{\egValCuda}{13.0}
\newcommand{\egValTorch}{2.12.1}
\newcommand{\egValRefN}{8192}

\newcommand{\egValTau}{\num{6.68e-04}}

\newcommand{\egValTriadGBs}{1578}

\newcommand{\egValGemmFpEightMedian}{519.1}
\newcommand{\egValGemmFpSixteenMedian}{217.3}
\newcommand{\egValScatterNd}{\num{1.66e-05}}
\newcommand{\egValIndexNd}{\num{1.29e-05}}
\newcommand{\egValCountBitStable}{6}

\newcommand{\egValPeakMem}{24.0}
\newcommand{\egValADevice}{RTX PRO 6000 Server}

\newcommand{\egValACardMemGiB}{95}
\newcommand{\egValASMCount}{188}

\newcommand{\egValADriver}{590.48.01}

\newcommand{\egValATorch}{2.12.1}

\newcommand{\egValATau}{\num{6.06e-04}}

\newcommand{\egValATriadGBs}{1486}

\newcommand{\egValAGemmFpEightMedian}{797.2}
\newcommand{\egValAGemmFpSixteenMedian}{404.2}
\newcommand{\egValAScatterNd}{\num{4.02e-06}}
\newcommand{\egValAIndexNd}{\num{4.01e-06}}
\newcommand{\egValACountBitStable}{6}

\newcommand{\egValAPeakMem}{66.0}
\newcommand{\egNumReplayDevices}{2}
\newcommand{\egReplayDevices}{RTX 5090 and RTX PRO 6000 Server}
\newcommand{\egValMaxFloorRatio}{1.26}
\newcommand{\egValMinFloorRatio}{1.00}
\newcommand{\egMerkleRoot}{8c19dd13728bd0de}
\newcommand{\egNumClaims}{70}
\newcommand{\egProofLen}{7}
\newcommand{\egContinuationLength}{2}
\newcommand{\egContinuationHead}{d1c61484bef20151}

\newcommand{\egContinuationServerFloorBig}{\num{1.99e-04}}
\newcommand{\egContinuationBigN}{32768}
\newcommand{\egResurrectPresent}{H200 NVL}
\newcommand{\egResurrectDead}{RTX PRO 6000 Blackwell Server Edition}
\newcommand{\egResurrectFraction}{100}
\newcommand{\egResurrectFloorRecovered}{20/20}
\newcommand{\egResurrectClassRecovered}{8/8}
\newcommand{\egResurrectLost}{8}
\newcommand{\egFingerprintKey}{b552f0f2c1e1e43c}
\newcommand{\egFingerprintDist}{1.3}
\newcommand{\egFingerprintTol}{10}
\newcommand{\egFingerprintBW}{7013}
\newcommand{\egFingerprintMem}{144}
\newcommand{\egBTwoHundredTput}{3324}

\newcommand{\egSeamRatio}{16}
\newcommand{\egSeamCacheRate}{662}
\newcommand{\egSeamHbmRate}{179}
\newcommand{\egSeamCrossRate}{43}
\newcommand{\egSeamPressureStable}{stable}

\newcommand{\egEvadeRel}{50}
\newcommand{\egEvadeRho}{\num{2.36e-04}}
\newcommand{\egEvadeWindowFpSixteen}{0.1}
\newcommand{\egEvadeWindowFpEight}{0.5}
\newcommand{\egFiatShamirRho}{\num{6.36e-01}}
\newcommand{\egWitnessRho}{1.43}
\newcommand{\egNumWitness}{4}
\newcommand{\egCommittedFooled}{4/4}
\newcommand{\egFsCaught}{4/4}
\newcommand{\egMwEvadeRel}{50}
\newcommand{\egStressGemms}{\num{139264}}
\newcommand{\egStressMinutes}{14}

\newcommand{\egStressPeakTemp}{80}
\newcommand{\egStressPeakPower}{588}
\newcommand{\egStressDidtPower}{537}
\newcommand{\egStressCap}{575}

\newcommand{\egStressInflationHot}{1.00}
\newcommand{\egStressInflationDidt}{1.00}
\newcommand{\egStressErrors}{0}

\title{\bfseries Self-Verifying Measurement Records:\\[2pt]
Hash-Linked Evidence Graphs for Hardware Benchmarking}

\author{%
  Faruk Alpay\thanks{Corresponding author: \texttt{alpay@lightcap.ai}.}\quad Bar\i{}\c{s} Ba\c{s}aran\\[2pt]
  \normalsize Department of Computer Engineering, Bah\c{c}e\c{s}ehir University, Istanbul, Turkey\\
  \normalsize \texttt{\{faruk.alpay, baris.basaran\}@bahcesehir.edu.tr}
}
\date{}

\begin{document}
\maketitle

\begin{abstract}
\noindent
Performance numbers reported for hardware are accepted on trust: the reader cannot recompute them,
the apparatus is gone, and the silicon itself can be silently wrong, with fleet studies reporting on
the order of one core in a thousand returning incorrect arithmetic with no error raised. We make a
reported hardware measurement a tamper-evident, independently checkable record. Every quantity that
appears in the text, a table, or a figure is bound, by its content hash, to the observation and the
verification behind it; the whole is a hash-linked, append-only structure, a transparency log for
measurement, that a verifier audits offline without trusting its producer. Matrix products are
verified by a probabilistic identity at $O(kn^2)$ cost under a tolerance we derive from
floating-point error analysis and calibrate to the device's own measured residual floor, so a wrong
product is rejected with probability $1-2^{-k}$; quantities with no such identity carry an algebraic
checksum and a measured reproducibility class, which together catch a fault that recurs identically
on every run. We then treat the check itself as a security object. A probe seed committed for offline
reproducibility is an attack surface: a probe-aware adversary hides a corruption of \egEvadeRel\% of
the output in the probe's null space, and a coordinated version passes a half-wrong result through
a quorum of \egNumWitness{} bit-identical witnesses at once; a Fiat-Shamir challenge derived from the
claimed output closes both while keeping the record offline. Driving the device from an unprivileged
tenant's reach, with a $di/dt$ power virus and a thermal soak, neither moves the calibrated tolerance
nor produces a silent error, which places the physical fault threat at the rare defective part or the
privileged attacker and marks the boundary at which the record must compose with a hardware root of
trust. We demonstrate the construction across Blackwell and Hopper accelerators, and report a measured residual-floor and reproducibility map by
precision, size, and device, including that the floor the tolerance rests on is an
architecture-independent invariant under which a retired device's record is reconstructed on a
successor of a different architecture, and a per-device signature that the part re-derives on demand.
\end{abstract}

\section{Introduction}
A reader who meets ``the kernel sustains \SI{1.7}{\peta\flop\per\second}'' cannot check it. The
number is the end of a computation whose inputs, the device, the driver, the exact shapes, and the
moment-to-moment clock, are gone by the time the sentence is read, and on shared, thermally varying
hardware the same run need not return the same value. Worse, the hardware can be wrong without
saying so: large-fleet studies find a small but steady fraction of cores, of order one in a
thousand, that produce incorrect results with no error raised~\cite{hochschild2021cores,dixit2021silent},
and accelerators are not exempt. A reported performance number is therefore detached from two things
at once: the evidence that it was measured, and the evidence that the underlying computation was
even correct.

The expectation that a stated number carry its evidence is no longer only a matter of taste.
Scientific data are asked to be findable and reusable in ways that let others check
them~\cite{wilkinson2016fair}; international climate accounting is built on a measure, report, and
verify spine in which a figure is only as good as the trail behind it~\cite{unfccc2015paris}; and
recent rules on high-impact computing systems require that records of what ran be
kept~\cite{eu2024aiact}. We take the same stance for a benchmark: a measurement is reported well
only when each number it states carries the evidence for itself, and that evidence should be
checkable by a reader holding nothing but the paper's own archive. We treat this as an integrity
problem rather than a presentational one. The record must hold not only against accidental error but
against a producer with reason to misreport, so we ask of it what is asked of a transparency log or a
verifiable computation: that a verifier need not trust the party that produced the
data~\cite{laurie2014ct,crosby2009tamper}, and we make the underlying hardware itself the witness and
the signer.

\paragraph{Approach.}
We attach to every reported quantity the content hash of the record that produced it, and we attach
to that record a verification chosen by what the quantity admits.
A \emph{linear} quantity, the entry of a matrix product, is checked by a probabilistic identity:
for a claimed $C=AB$ and a random probe $x$, the residual $\lVert A(Bx)-Cx\rVert$ is small only if
$C$ is correct, and a wrong $C$ is rejected with probability at least $1-2^{-k}$ over $k$ probes, at
$O(kn^2)$ cost rather than recomputing the $O(n^3)$ product~\cite{freivalds1977}. What makes this a
measurement rather than a textbook check is the tolerance: we derive it from the rounding-error
analysis of the product~\cite{higham2002accuracy,higham2019probabilistic} and calibrate it to the
device's own measured residual floor, so that ordinary low-precision roundoff is admitted while a
genuine error is not. A \emph{nonlinear} quantity, a reduction or an attention output, has no such
cheap algebraic check, algorithm-based fault tolerance is known not to extend to it~\cite{huang1984abft},
so its record instead carries the measured run-to-run reproducibility of the kernel: bit-stable, or
bounded by a measured divergence. The reported throughput itself carries its own class: a tightly
bounded rate, or one that tracks the device's clock and thermal state.

Observations, verifications, reductions, the displayed quantities, and the compiled document form
one hash-linked acyclic graph (Figure~\ref{fig:graph}). Verifying a number is a local act: re-hash
the nodes on its path and recompute the small functions that join them. The whole graph is audited
in a single offline pass over one archive, with no device and no network (Section~\ref{sec:graph}).

\paragraph{Multi-stage verification.}
Verification is not one decision but a short transcript the record commits to, stage by stage:
acquire a result, witness its residual and the calibrated tolerance, decide, and on a flagged result
repair by recomputing on a higher-assurance path and re-witness (Section~\ref{sec:linear}). We show
two instances on the instrument. An FP8 product's residual exceeds a tolerance calibrated at FP16;
the record flags it and a higher-precision recomputation brings it back within tolerance. A single
bit flipped in a correct product drives the residual far past tolerance; the identity check rejects
it with the stated probability, and a recomputation repairs it. These are the loop of an error the
data reveals, a fix, and a re-check, recorded in full.

\paragraph{Cross-device corroboration.}
With two cards of the same model the record gains an independent check. The two devices compute the
same product from identical inputs, and the difference between their outputs is itself a measured
quantity: writing $\mathcal{D}[q]=q(\text{card}_1)-q(\text{card}_0)$ for the difference of a quantity
across the two parts, we find $\mathcal{D}$ of the output is exactly zero for the deterministic
kernels, so each card corroborates the other bit for bit, while $\mathcal{D}$ of the throughput is a
small but consistent nonzero signature that two identical parts still carry (Section~\ref{sec:cross}).
A fault on one card breaks the agreement and is caught by the disagreement, the oldest form of a
check, two witnesses, here measured rather than assumed.

\paragraph{Hardware platform.}
We run on two \egDevice{} cards, drive them to their documented limits in bandwidth, board power, and
resident memory, and then measure what is not documented: the residual floor of their matrix products
as a function of precision and size, which is exactly the quantity the tolerance is calibrated
against. Blackwell has no detailed public timing model, and recent work characterizes it only through
microbenchmarks~\cite{blackwell2025dissect,blackwell2025microbench}; the floor and the per-card
signature we report are of that kind. To check that we have not merely described one workstation, we
also replay the same checks on a single \egValDevice{} with \egValCardMemGiB{}\,GiB memory and a
single \egValADevice{} with \egValACardMemGiB{}\,GiB, spanning consumer and server product
configurations~\cite{nvidia2026rtx5090,nvidia2026rtxpro6000server}. Nothing in the method is specific
to these parts, and the same record could be built for any accelerator that computes the same kinds
of quantity, down to a future device whose internals differ entirely from a GPU. The point of a
self-verifying record is finally a human one: a reader should trust a reported number to exactly the
degree its evidence warrants, which is the calibration that separates appropriate reliance from
misplaced faith~\cite{lee2004trust}, and which a number with no attached evidence cannot support.

\paragraph{Contributions.}
\begin{itemize}\itemsep2pt
  \item A hash-linked evidence graph binding every displayed hardware-benchmark quantity to its
  observation and its verification, audited in one offline pass; tamper-evident under collision
  resistance, linear-time checkable (Section~\ref{sec:graph}).
  \item A verification for linear quantities: an identity check with a tolerance derived from
  rounding-error analysis and calibrated to the device's measured residual floor, with a detection
  guarantee, recorded as a multi-stage transcript that reveals, repairs, and re-checks
  (Section~\ref{sec:linear}).
  \item A measured account of what the instrument admits: the residual floor of Blackwell matrix
  products by precision and size, and a reproducibility class for every workload, including atomic
  kernels whose run-to-run output we find nondeterministic by a measured margin
  (Section~\ref{sec:repro}).
  \item An algebraic check for the nonlinear quantities: attention decomposed into its two matrix
  products and a softmax invariant, atomic accumulation checked by a weighted checksum, so a
  consistent fault that run-to-run divergence cannot see is still caught (Section~\ref{sec:nonlinear}).
  \item A two-device differential: identical-model cards that agree bit for bit on deterministic
  output, disagree under an injected fault, and carry a small consistent throughput signature, which
  turns a second card into an independent witness and a per-part fingerprint (Section~\ref{sec:cross}).
  \item Two out-of-sample single-device replays, on \egReplayDevices{}: the same seeded checks across
  consumer and server configurations, used to separate method-level invariants from device-specific
  performance (Section~\ref{sec:validation}).
  \item Re-verification that is portable and sub-linear: inputs from committed seeds let a reader
  re-check on any device, including a CPU, and a Merkle commitment over the claims audits any single
  number in $O(\log N)$ (Section~\ref{sec:archive}).
  \item A directed continuation across devices: each instrument's record commits to the hash of the
  record it extends, so the evidence advances in one direction, a later device only appends, and an
  earlier head stays valid as the record grows (Section~\ref{sec:validation}).
  \item What survives when an instrument is gone: from a retired device's seeds, a living one of a
  different architecture re-derives \egResurrectFraction\% of its record, its transferable invariants,
  while its device-local rates are lost; a recallable signature the device answers for itself; and a
  random-access reading that exposes a \egSeamRatio$\times$ cross-die penalty the part's streaming
  numbers hide (Section~\ref{sec:gone}).
  \item A security analysis of the check itself: a committed probe lets a probe-aware adversary pass a
  product wrong by \egEvadeRel\% of its norm, and a coordinated version fools a quorum of
  \egNumWitness{} bit-identical witnesses at once, which a Fiat-Shamir challenge derived from the
  output closes on every witness while keeping the record offline and reproducible (Section~\ref{sec:attack}).
  \item A complete archive, the code, the \egNumObservations{} observations, the verification
  transcript, the graph, and a hash manifest, verifying to the root \hx{\egRootHash{}\ldots} with a
  standard-library checker (Section~\ref{sec:archive}).
\end{itemize}

\FloatBarrier
\section{A hash-linked evidence graph}\label{sec:graph}

\subsection{Nodes and content addressing}
Let $H$ be SHA-256~\cite{nist2015sha} over a canonical byte image: a JSON object with keys in
lexicographic order, compact separators, UTF-8 text, every float written as the shortest decimal
that round-trips to the same IEEE-754 double, and a terminating newline. Canonicalization gives a
value exactly one byte image, so two parties with the same content compute the same identifier.

\begin{definition}[Evidence graph]
An \emph{evidence graph} $G=(V,E)$ has four kinds of vertex.
\emph{(i)} An \emph{observation} $o$ is the canonical record of one measurement: a workload or a
verification step, its configuration, the environment digest, the device covariates sampled during
it, the numerical fingerprint of the output, and the residual where one is computed;
$\mathrm{id}(o)=H(\mathrm{bytes}(o))$.
\emph{(ii)} A \emph{reduction} $r=(f,\theta,[\mathrm{id}(u_i)],v)$ applies a named pure function $f$
to inputs named only by their identifiers and commits to $v=f(u_1,\dots;\theta)$.
\emph{(iii)} A \emph{claim} $c$ is a quantity as rendered in the document: a value, a unit, a
tolerance, a class, and the one identifier it asserts.
\emph{(iv)} The \emph{root} names the claim identifiers the document displays. Each edge
$u\!\to\!v$ stores $\mathrm{id}(v)$; $G$ is acyclic with the layering
$\text{root}\to\text{claims}\to\text{reductions}\to\text{observations}$.
\end{definition}

Figure~\ref{fig:graph} shows the graph for one workload: observations enter on the left, a fixed
set of reductions reads them by hash and produces the median rate, its dispersion, the rank
correlation of rate against clock, the run-to-run divergence, and the verification residual; the
claims in the prose name those reductions; the document seals the claims. Every arrow is a hash.

\begin{figure}[H]
\centering
\begin{tikzpicture}[
  font=\small,
  >={Stealth[length=2.1mm]},
  obs/.style={draw=obsline, fill=obsfill, rounded corners=2pt, minimum width=23mm,
              minimum height=6.4mm, blur shadow={shadow blur steps=5, shadow xshift=0pt,
              shadow yshift=-1pt, shadow opacity=28}},
  red/.style={draw=redline, fill=redfill, rounded corners=2pt, minimum width=29mm,
              minimum height=6.2mm, blur shadow={shadow blur steps=5, shadow yshift=-1pt,
              shadow opacity=28}},
  clm/.style={draw=clmline, fill=clmfill, rounded corners=2pt, minimum width=31mm,
              minimum height=7mm, align=center, blur shadow={shadow blur steps=5,
              shadow yshift=-1pt, shadow opacity=28}},
  doc/.style={draw=docline, fill=docfill, rounded corners=2pt, minimum width=23mm,
              minimum height=11mm, align=center, blur shadow={shadow blur steps=5,
              shadow yshift=-1pt, shadow opacity=28}},
  hedge/.style={->, draw=soft, thick},
  hlab/.style={font=\scriptsize, soft, fill=white, inner sep=1pt},
]
\node[soft] at (0,2.7) {\bfseries\footnotesize observations};
\node[soft] at (4.1,2.7) {\bfseries\footnotesize reductions};
\node[soft] at (8.3,2.7) {\bfseries\footnotesize claims};
\node[soft] at (11.7,2.7) {\bfseries\footnotesize document};

\node[obs] (o1) at (0,1.8) {$o_1$};
\node[obs] (o2) at (0,0.95) {$o_2$};
\node[obs] (o3) at (0,0.1) {$o_3$};
\node[soft] (od) at (0,-0.6) {$\vdots$};
\node[obs] (oR) at (0,-1.35) {$o_{R}$};
\node[soft, font=\scriptsize] at (0,-2.05) {\egMaxRepeats{} repeats};

\node[red] (med) at (4.1,1.7) {$\Sigma$~~median rate};
\node[red] (disp) at (4.1,0.75) {$\Sigma$~~dispersion, $\rho_{\mathrm{clk}}$};
\node[red] (ndr) at (4.1,-0.2) {$\Sigma$~~run-to-run div.};
\node[red] (fr) at (4.1,-1.15) {$\Sigma$~~Freivalds resid.};

\node[clm] (cp) at (8.3,1.15) {\emph{rate} $=$ \egGemmBfOneSixMedian~T\\[-1pt]\textcolor{cS2}{\class{S2}}};
\node[clm] (cn) at (8.3,-0.95) {\emph{verified}\\[-1pt]\textcolor{cS0}{\class{S0}}, $\rho\!\le\!\tau$};

\node[doc] (doc) at (11.7,0.1) {document\\[2pt]\code{main.pdf}};

\foreach \o in {o1,o2,o3,oR}{
  \draw[hedge] (\o.east) -- (med.west);
  \draw[hedge, draw=soft!55] (\o.east) -- (fr.west);
}
\draw[hedge, draw=soft!55] (o2.east) -- (disp.west);
\draw[hedge, draw=soft!55] (o3.east) -- (ndr.west);
\draw[hedge] (med.east) -- (cp.west);
\draw[hedge, draw=soft!55] (disp.east) -- (cp.west);
\draw[hedge, draw=soft!55] (ndr.east) -- (cn.west);
\draw[hedge] (fr.east) -- (cn.west);
\draw[hedge] (cp.east) -- (doc.west);
\draw[hedge] (cn.east) -- (doc.west);

\node[hlab] at (1.95,2.05) {\hx{9f3a\ldots}};
\node[hlab] at (6.15,1.92) {\hx{c41e\ldots}};
\node[hlab] at (10.55,1.05) {\hx{7b02\ldots}};
\end{tikzpicture}
\caption{The evidence graph for one workload. Each arrow carries the content hash of the node it
points at (one prefix shown per layer). A reader verifies a displayed number by re-hashing the
observations, recomputing the reductions, and checking the claim against the value it names.}
\label{fig:graph}
\end{figure}
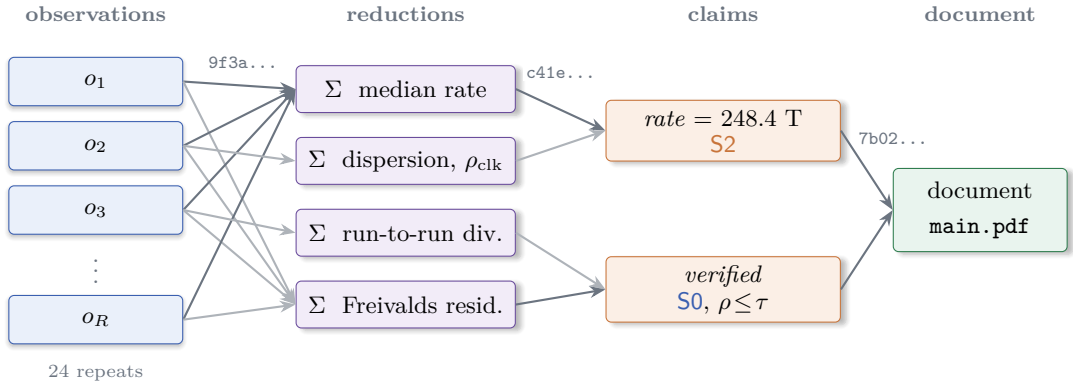

\subsection{Offline audit}
\begin{definition}[Audit pass]
Given the archived graph, the audit \emph{(a)} recomputes $\mathrm{id}(x)$ for every node and checks
it against the stored identifier; \emph{(b)} checks every edge points at a present node;
\emph{(c)} recomputes each reduction from its named inputs and checks the committed value.
\end{definition}

\begin{lemma}[Tamper-evidence]\label{lem:tamper}
If any byte of any node is altered, the audit fails unless an explicit SHA-256 collision is exhibited.
\end{lemma}
\begin{proof}
Altering a node changes its recomputed identifier, so step (a) fails for it unless the altered and
original bytes share a hash. Editing a referrer to match an altered target changes the referrer's
identifier, and the failure propagates upward to the root. Each step needs a preimage or collision
to avoid detection.
\end{proof}

\begin{theorem}[Single-pass offline audit]\label{thm:single}
The graph is auditable in time linear in its serialized size using only the archive, with no device,
network, or re-execution.
\end{theorem}
\begin{proof}
Each node is hashed once; each reduction reads inputs already present and applies a pure function
linear in their number. The reductions are medians, dispersions, rank correlations, divergences, and
residual aggregates, all pure functions of archived values; none consults the device.
\end{proof}

A reduction commits to its value, so the auditor recomputes rather than trusts. The reductions are
deliberately simple and use only the standard library, so the check of Section~\ref{sec:archive}
depends on nothing a reader must install.

\section{Verifying linear quantities}\label{sec:linear}

A matrix product is the heaviest quantity these workloads report and the one that admits a cheap
exact-in-expectation check. For a claimed $C=AB$ with $A,B,C\in\mathbb{R}^{n\times n}$ and a probe
matrix $X\in\mathbb{R}^{n\times k}$ with independent standard-normal columns, write the normalized
residual
\[
  \rho(C) \;=\; \frac{\bigl\lVert A(BX)-CX\bigr\rVert_\infty}{\lVert CX\rVert_\infty+\varepsilon},
\]
formed in float32 so the check is more accurate than the products it inspects, at $O(kn^2)$ cost.

\begin{lemma}[Identity check, after Freivalds]\label{lem:freivalds}
If $AB\neq C$ then a single random probe satisfies $A(Bx)=Cx$ with probability at most $\tfrac12$,
so over $k$ independent probes the disagreement is missed with probability at most
$2^{-k}$~\cite{freivalds1977}. If $AB=C$ the residual is zero in exact arithmetic.
\end{lemma}

In floating point the residual of a \emph{correct} product is not zero but a roundoff term. Standard
analysis bounds it by $\gamma_n=nu/(1-nu)$ times the operand norms, and the probabilistic refinement
replaces $\gamma_n$ by a constant growing like $\sqrt{n\log n}\,u$~\cite{higham2002accuracy,higham2019probabilistic},
where $u$ is the unit roundoff of the product's precision. We do not need the constant in closed
form; we measure it.

\begin{definition}[Device-calibrated tolerance]\label{def:tau}
Let $\bar\rho$ be the largest residual observed over the correct repeats of a kernel at a given
precision and size. The acceptance tolerance is $\tau=m\,\bar\rho$ for a fixed margin $m=\egVMargin$.
A claim is \emph{verified} iff $\rho\le\tau$.
\end{definition}

The floor $\bar\rho$ is a property of the silicon and the precision, not of the algorithm, and
Section~\ref{sec:repro} reports it for Blackwell. Calibrating to it is what separates legitimate
low-precision roundoff from a real error: a precision whose floor is high gets a correspondingly
loose tolerance, and a residual that exceeds even that loose bound is not roundoff.

\begin{proposition}[Calibrated detection]\label{prop:detect}
Fix a precision with measured floor $\bar\rho$ and tolerance $\tau=m\bar\rho$. A correct product is
accepted, since $\rho\le\bar\rho\le\tau$. A product whose deviation from $AB$ exceeds $\tau$ in a
probed coordinate is rejected with probability at least $1-2^{-k}$. Hence a verified linear claim
equals $AB$ to within $\tau$ with confidence $1-2^{-k}$.
\end{proposition}
\begin{proof}
Acceptance is $\rho\le\tau$; the first statement is the calibration. For the second, write
$D=C-AB$; the check sees $DX$ up to the float32 roundoff of the witness, which is below $\bar\rho\le\tau$.
If $D$ deviates by more than $\tau$ in a probed coordinate, $\lVert DX\rVert$ exceeds the witness
roundoff for at least one column unless every column lies in the same measure-zero subspace, which by
Lemma~\ref{lem:freivalds} happens with probability at most $2^{-k}$ over the $k$ independent columns.
\end{proof}

\paragraph{The verification transcript.}
A verification is recorded as a short sequence of stages, each a node in the graph: \emph{acquire}
the device's result; \emph{witness} its residual $\rho$ and the tolerance $\tau$; \emph{decide}
(accept if $\rho\le\tau$); and, on a flagged result, \emph{repair} by recomputing on a
higher-assurance path and re-witness. The auditor of Section~\ref{sec:graph} checks that every
decision follows its rule and every stage links to the next by hash, so the transcript is itself
tamper-evident. Table~\ref{tab:verify} is the transcript of two demonstrations on the instrument.

\begin{table}[H]
\centering
\caption{Two verification transcripts on the instrument, calibrated at the FP16 floor with
$\tau=\egVtau$ and $k=\egK$ probes (miss probability $\le\egPmiss$). \emph{Precision}: an FP8 product
exceeds the high-precision tolerance and is repaired by recomputing with float32 accumulation.
\emph{Corruption}: a single injected bit flip in a correct FP16 product is rejected and repaired.}
\label{tab:verify}
\small
\begin{tabular}{lllrl}
\toprule
Case & Stage & Operation & Residual & Decision \\
\midrule
precision & acquire & FP8 product & \num{1.58e-03} & flag \\
 & repair & FP32-accumulation recompute & \num{2.74e-04} & accept \\
\midrule
corruption & acquire & FP16 product & \num{3.06e-04} & accept \\
 & inject & single bit flip in C & \num{2.01e-03} & reject \\
 & repair & recompute C & \num{3.06e-04} & accept \\
\bottomrule
\end{tabular}

\end{table}

The corruption row is the point of the construction: a residual of \egFlipResid{}, against a
tolerance of \egVtau, is not a question of roundoff, and Proposition~\ref{prop:detect} makes its
rejection a statement with a probability attached rather than a heuristic. The bit flip is injected
deliberately, to measure the detector's sensitivity rather than to claim such a deviation is common;
how often one arises on its own, and whether a tenant can physically provoke one, is taken up in
Section~\ref{sec:physical}. The precision row shows the same machinery distinguishing an underprecise
result, accepted by no high-precision reader, from a correct one, and recording the repair that fixes
it.

\section{Verifying nonlinear quantities}\label{sec:nonlinear}

The identity check covers matrix products. The other heavy quantities, an attention output (a softmax
over two matrix products) or an atomic accumulation, are nonlinear, and an algebraic fault-tolerance
check is known not to extend to them directly~\cite{huang1984abft}. Leaving them to their run-to-run reproducibility is
not enough: that detects only \emph{non}-reproducibility, so a fault that occurs identically on every
run, the kind a deterministic but wrong kernel produces, leaves the divergence at zero and passes. We
close this with a check that compares against meaning rather than against a previous run.

\paragraph{Attention.} An attention output is $O=\mathrm{softmax}(QK^{\top}/\sqrt{d})\,V$, two
matrix products around one softmax. The verifier forms the scores and the probabilities in float32,
checks the softmax invariant that each probability row sums to one (we measure a worst-case row-sum
error of \egAttRowsum), and compares the reference output to the claim. The two products are
themselves Freivalds-checkable; the softmax is the only genuinely nonlinear step, and it is pinned by
its invariant. This follows the recent line that brings algorithm-based fault tolerance to attention
layers~\cite{flashabft2025,fttransformer2025}, here as a recorded check rather than an inline one.

\paragraph{Atomic accumulation.} A scatter or index-add is a sum into buckets, and admits a cheap
one-pass check: for random weights $w$, $\sum_j w_j\,\mathrm{out}_j$ must equal
$\sum_i w_{\mathrm{idx}_i}\,\mathrm{src}_i$, the additive analogue of the Freivalds identity. It costs
a single pass rather than re-running the scatter.

\paragraph{A consistent fault is caught.} Table~\ref{tab:nonlinear} is the point. We inject a fault
that is identical on every run, so its run-to-run divergence is exactly zero and a reproducibility
check sees nothing wrong. The attention residual rises from a floor of \egAttFloor{} to \egAttFault{},
and the scatter residual from \egScatterFloor{} to \egScatterFault{}: both far above their floors, so
the algebraic check catches what the divergence misses. The \class{Snd} quantities of
Section~\ref{sec:repro} therefore carry an algebraic guard, not only a measured band.

\begin{table}[H]
\centering
\caption{An injected fault that is identical on every run leaves the run-to-run divergence at zero,
so a reproducibility check passes it; the algebraic residual rises far above its floor and catches it.}
\label{tab:nonlinear}
\small
\begin{tabular}{lrrrc}
\toprule
Operation & Residual floor & Consistent-fault & Run-to-run div. & Caught \\
\midrule
attention & \num{2.00e-03} & \num{2.22e-01} & 0 & \checkmark \\
scatter & \num{1.03e-07} & \num{7.95e-04} & 0 & \checkmark \\
\bottomrule
\end{tabular}

\end{table}

\section{Reproducibility classes and residual floors}\label{sec:repro}

\paragraph{Device and protocol.}
Measurements are on two \egDevice{} cards (\egSMCount{} SMs, \egCardMemGiB{}\,GiB GDDR7 each, compute
capability \egComputeCap{}), driver \egDriver{}, CUDA~\egCuda{}, PyTorch~\egTorch{}; the workload grid
runs on one card and the differential of Section~\ref{sec:cross} on both. Each workload runs a
sustained inner loop repeated \egMaxRepeats{} times under load while the SM clock, temperature, and
power are sampled every \SI{20}{\milli\second}. The part is documented with a high-bandwidth GDDR7
interface and a \SI{300}{\watt} board limit~\cite{nvidia2026rtxpro6000maxq}; we hold FP8 \textsc{gemm}
at \egGemmFpEightMedian~TFLOP/s under the board-power limit (Table~\ref{tab:taxonomy}), and a memory-bound
stream fills \egFillMem~GiB of each card (Section~\ref{sec:cross}).

\paragraph{The residual floor.}
The quantity that calibrates the tolerance is the Freivalds residual of a \emph{correct} product.
At $n=\egFloorSizes$, the measured sequences are \egFloorSeqFpSixteen{} for FP16,
\egFloorSeqTfThreeTwo{} for TF32, \egFloorSeqBfSixteen{} for BF16, and \egFloorSeqFpEight{} for
FP8, in increasing-size order. Each precision stays in a narrow, size-dependent band rather than
moving monotonically, while the ordering is stable: FP16 and TF32 remain below BF16 and FP8. To our
knowledge this floor has not been reported for Blackwell; it is the empirical content of
$\bar\rho$ in Definition~\ref{def:tau}, and it is why a tolerance calibrated at FP16 flags an FP8
product (Table~\ref{tab:verify}).

\paragraph{Reproducibility of nonlinear quantities.}
Where no cheap identity exists, the record carries the measured run-to-run reproducibility instead.
Table~\ref{tab:taxonomy} gives, for every workload, the numerical class, bit-stable across repeats
(\class{S0}) or bounded by a measured relative divergence (\class{Snd}), and the performance class,
tightly bounded (\class{S1}) or tracking the device clock (\class{S2}). The dense \textsc{gemm}
kernels, at all four precisions, and the streaming and reduction kernels return bit-identical output
across every repeat; the atomic scatter and index-add kernels do not, diverging run to run by
\egScatterNd{} and \egIndexNd{} in relative terms. This is the boundary that an algebraic
fault-tolerance check cannot draw~\cite{huang1984abft} and that the FP8 reproducibility question
leaves open~\cite{shanmugavelu2024fpna}: on this part, low-precision dense products are bit-stable,
while atomic accumulation is the kernel class that is not. A \class{Snd} claim carries its measured
divergence as its tolerance, exactly as a verified linear claim carries $\tau$, and, by
Section~\ref{sec:nonlinear}, also an algebraic check that a consistent fault cannot slip past.

\begin{table}[H]
\centering
\caption{What each workload admits on a single \egDevice{}. \emph{Median} is the sustained rate;
\emph{Disp.} the relative MAD over \egMaxRepeats{} repeats; $\rho_{\mathrm{clk}}$ the Spearman
correlation of rate against SM clock; \emph{Perf.} its class (\class{S1} bounded, \class{S2}
clock-tracking); \emph{Output} the numerical class (\class{S0} bit-stable, \class{Snd} bounded by the
measured divergence). \emph{Resid.} is the Freivalds residual for the linear (\textsc{gemm})
workloads; the final columns record median board power, temperature, and peak resident memory.}
\label{tab:taxonomy}
\small
\begin{tabular}{llrrrrccrrr}
\toprule
 & & \multicolumn{4}{c}{Reported quantity} & \multicolumn{2}{c}{Numerics} & \multicolumn{3}{c}{Device state} \\
\cmidrule(lr){3-6}\cmidrule(lr){7-8}\cmidrule(l){9-11}
Workload & Var. & Median & Disp.\,(\%) & $\rho_{\mathrm{clk}}$ & Perf. & Output & Resid. & W & $^\circ$C & GiB \\
\midrule
gemm & fp16 & 242.4\,T & 0.188 & +0.80 & \class{S2} & \class{S0} & \num{2.02e-04} & 300 & 90 & 4.5 \\
gemm & bf16 & 248.4\,T & 0.055 & +0.48 & \class{S1} & \class{S0} & \num{1.69e-03} & 300 & 89 & 4.5 \\
gemm & tf32 & 117.7\,T & 0.063 & -0.31 & \class{S1} & \class{S0} & \num{2.89e-04} & 300 & 88 & 4.0 \\
gemm & fp8 & 484.2\,T & 0.022 & -0.12 & \class{S1} & \class{S0} & \num{1.71e-03} & 300 & 89 & 4.0 \\
hbm\_triad & f32 & 1510.4\,G & 0.019 & +0.41 & \class{S2} & \class{S0} & -- & 300 & 85 & 60.0 \\
attention & bf16 & 224.9\,T & 0.596 & +0.29 & \class{S1} & \class{S0} & -- & 300 & 91 & 1.3 \\
reduction & default & 1635.4\,G & 0.004 & -0.07 & \class{S1} & \class{S0} & -- & 300 & 85 & 8.0 \\
reduction & deterministic & 1635.2\,G & 0.013 & +0.19 & \class{S1} & \class{S0} & -- & 300 & 85 & 8.0 \\
scatter\_add & atomic & 73.0\,Ge & 0.510 & +0.32 & \class{S1} & \class{Snd} & -- & 300 & 86 & 12.0 \\
index\_add & atomic & 72.8\,Ge & 0.073 & +0.44 & \class{S1} & \class{Snd} & -- & 299 & 87 & 12.0 \\
\bottomrule
\end{tabular}

\end{table}

\FloatBarrier
Figure~\ref{fig:drift} closes the performance side: throughput is flat against clock for a bounded
workload and sloped for one that tracks the boost clock, which is what the rank correlation in
Table~\ref{tab:taxonomy} reports, so the class is read from the covariate, not chosen.

\begin{figure}[H]
\centering
\includegraphics[width=0.80\linewidth]{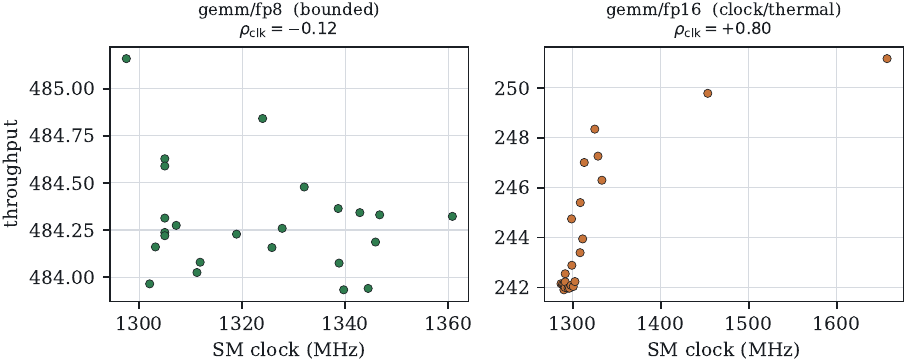}
\caption{Throughput against SM clock for a bounded (\class{S1}) and a clock-tracking (\class{S2})
workload. The slope is the evidence behind the class.}
\label{fig:drift}
\end{figure}

\section{Cross-device differential measurement}\label{sec:cross}

A single card witnesses its own product through the residual of Section~\ref{sec:linear}. A second
card of the same model witnesses it again, independently, and the difference between the two is a
measurement in its own right. For a quantity $q$ computed on each part write the device differential
\[
  \mathcal{D}[q] \;=\; q(\text{card}_1) - q(\text{card}_0),
\]
a finite difference along the one axis a benchmark usually holds fixed, the identity of the silicon.
We run both \egDevice{} cards (\egCardMemGiB{}\,GiB each) on the same inputs and read three differentials.

\paragraph{Output agreement.}
For every deterministic kernel, both cards return bit-identical results: the relative difference
$\mathcal{D}[\,\text{output}\,]$ of their matrix products is \egCrossDelta{} at every size and
precision (Table~\ref{tab:cross}). Two parts agreeing to the last bit is the strongest corroboration
a linear claim can have, an exact cross-check that needs no tolerance, and it holds because the dense
products are bit-stable (Section~\ref{sec:repro}) on both parts at once. The residual floors agree as
well: $\bar\rho$ is the same on the two cards, so the floor is a property of the architecture, not of
the individual die.

\paragraph{Throughput divergence.}
The same two cards do not run at the same speed. Their FP16 \textsc{gemm} throughputs differ by
\egTputDiffFpSixteen\% and their FP8 by \egTputDiffFpEight\%, consistently in the same direction, a
small fixed signature that identical models carry from manufacturing and die binning and that holds
steady across runs~\cite{sinha2022notall}. $\mathcal{D}[\,\text{throughput}\,]$ is thus a per-part
fingerprint sitting beside an output differential of zero: the cards are numerically the same machine
and physically not, and a record that reports a rate without naming the card has already lost
information the differential makes visible. Driven to fill their memory, each card holds
\egFillMem{}\,GiB resident while a memory-bound stream sustains \egFillTput{}~GB/s.

\paragraph{Fault detection by disagreement.}
When the two outputs agree bit for bit, a disagreement can only come from one side going wrong. We
inject a single bit flip into one card's product; the output differential jumps from \egCrossDelta{}
to \egWitnessDeltaCorrupt{}, and the corruption is caught by the discrepancy alone, while the other
card, and its Freivalds residual, corroborate the correct result. This is redundancy used as
evidence: a second witness that is silent when both are right and loud when one is not, recorded in
the same graph as everything else.

\begin{table}[H]
\centering
\caption{The differential between two \egDevice{} cards on identical inputs. $\rho_0,\rho_1$ are each
card's residual floor; $\mathcal{D}[\,\text{out}\,]$ is the relative difference of their outputs,
zero wherever the kernel is deterministic.}
\label{tab:cross}
\small
\begin{tabular}{lrrrr}
\toprule
Precision & $n$ & $\rho_0$ & $\rho_1$ & $\mathcal{D}[\mathrm{out}]$ \\
\midrule
fp16 & 4096 & \num{2.02e-04} & \num{2.02e-04} & 0 \\
fp16 & 8192 & \num{1.93e-04} & \num{1.93e-04} & 0 \\
fp16 & 16384 & \num{2.02e-04} & \num{2.02e-04} & 0 \\
bf16 & 4096 & \num{1.78e-03} & \num{1.78e-03} & 0 \\
bf16 & 8192 & \num{1.53e-03} & \num{1.53e-03} & 0 \\
bf16 & 16384 & \num{1.69e-03} & \num{1.69e-03} & 0 \\
tf32 & 4096 & \num{3.56e-04} & \num{3.56e-04} & 0 \\
tf32 & 8192 & \num{2.86e-04} & \num{2.86e-04} & 0 \\
tf32 & 16384 & \num{2.89e-04} & \num{2.89e-04} & 0 \\
fp8 & 4096 & \num{1.78e-03} & \num{1.78e-03} & 0 \\
fp8 & 8192 & \num{1.61e-03} & \num{1.61e-03} & 0 \\
fp8 & 16384 & \num{1.71e-03} & \num{1.71e-03} & 0 \\
\bottomrule
\end{tabular}

\end{table}

\FloatBarrier

\section{Out-of-sample device validation}\label{sec:validation}

The two-card result above is a differential inside one SKU. A stronger check changes the product
configuration itself. We therefore replay the same seeded products, residual witnesses, and
reproducibility measurements on two single-card instruments. The first is an \egValDevice{}:
\egValSMCount{} SMs, \egValCardMemGiB{}\,GiB of usable memory, driver \egValDriver{}, and
PyTorch~\egValTorch{}~\cite{nvidia2026rtx5090}. The second is an \egValADevice{}:
\egValASMCount{} SMs, \egValACardMemGiB{}\,GiB, driver \egValADriver{}, and
PyTorch~\egValATorch{}~\cite{nvidia2026rtxpro6000server}. Both report compute capability
\egValComputeCap{} and CUDA~\egValCuda{}, but they span consumer and server power/memory
configurations. Equality of rates is therefore neither expected nor the claim. The transfer question
is whether the \emph{record} keeps the same semantics when the instrument changes.

For device $d$, precision $p$, and size $n$, write
\[
  F_d(p,n)=\rho_d(p,n), \qquad
  K_d(w)=(\text{performance class},\text{numerical class})
\]
for the residual floor and the class pair assigned to workload $w$. A device replay is
evidence-preserving when the same pure reductions recompute, the linear claims use the same rule
$\rho_d\le mF_d(\mathrm{fp16},n_{\mathrm{ref}})$ with a fresh device floor, and the numerical class
boundary $K_d$ changes only where the measured divergence changes. This is weaker than bit-exact
cross-device replay and stronger than a speed comparison: it asks whether the evidence language has
the same meaning after the instrument changes.

\begin{table}[H]
\centering
\caption{Out-of-sample residual floors on \egReplayDevices{} compared with the primary \egDevice{}
instrument at the same reference size. Parentheses give each replay floor divided by the primary
floor; each tolerance is recalibrated on its own device rather than copied across.}
\label{tab:validation}
\small
\begin{tabular}{lrrrr}
\toprule
Precision & $n$ & Max-Q floor & RTX 5090 floor & RTX PRO 6000 Server floor \\
\midrule
fp16 & 8192 & \num{1.93e-04} & \num{2.23e-04} (1.15$\times$) & \num{1.93e-04} (1.00$\times$) \\
bf16 & 8192 & \num{1.53e-03} & \num{1.72e-03} (1.12$\times$) & \num{1.53e-03} (1.00$\times$) \\
tf32 & 8192 & \num{2.86e-04} & \num{3.60e-04} (1.26$\times$) & \num{2.86e-04} (1.00$\times$) \\
fp8 & 8192 & \num{1.61e-03} & \num{1.75e-03} (1.08$\times$) & \num{1.61e-03} (1.00$\times$) \\
\bottomrule
\end{tabular}

\end{table}

Table~\ref{tab:validation} gives the transfer test. At $n=\egValRefN$, all replay residual floors
stay within a \egValMinFloorRatio--\egValMaxFloorRatio{} multiplicative band of the primary
instrument across FP16, BF16, TF32, and FP8. The independently calibrated tolerances are
\egValTau{} on the \egValDevice{} and \egValATau{} on the \egValADevice{}. The numerical boundary
also transfers twice: the replays have \egValCountBitStable{} and \egValACountBitStable{}
bit-stable workloads, respectively, while scatter-add and index-add remain \class{Snd}. Their
measured divergences are \egValScatterNd{}/\egValIndexNd{} on the first device and
\egValAScatterNd{}/\egValAIndexNd{} on the second.

Performance is the part that does \emph{not} transfer. On the common validation shapes, the
\egValDevice{} sustains \egValGemmFpSixteenMedian{} and \egValGemmFpEightMedian{}~TFLOP/s, whereas
the \egValADevice{} sustains \egValAGemmFpSixteenMedian{} and
\egValAGemmFpEightMedian{}~TFLOP/s. Their memory triads reach \egValTriadGBs{} and
\egValATriadGBs{}~GB/s, respectively, with \egValPeakMem{} and \egValAPeakMem{}\,GiB peak resident
sets imposed by the replay shape. The added server result strengthens the separation: residual-floor order and
the deterministic/nondeterministic boundary survive two changes of instrument, while throughput
remains a device-local covariate.

\paragraph{Continuation across devices.} A replay is not only a test but an extension. Each device's
record commits to the hash of the record it continues, so the second instrument's record folds in the
first, and the evidence advances in one direction rather than sitting as a loose pile of independent
files. The genesis is the primary record; the \egValDevice{} extends it, and the \egValADevice{}
extends that, leaving a head of \hx{\egContinuationHead{}\ldots} after \egContinuationLength{} links.
Because a link is the content hash of (the prior head, the device's summary), the head is recomputed
by anyone holding the records, a later device only appends, and an earlier head stays valid as the
record grows; a device may also continue from any prior head, so the record branches when two
instruments extend the same point. The server link was produced by driving that card to
\egContinuationBigN{}-wide products at its board-power limit, where its FP16 residual floor reads
\egContinuationServerFloorBig{}, in line with the smaller sizes and extending the floor measurement
upward on a third configuration.

\FloatBarrier

\section{Reconstruction from committed seeds and device signatures}\label{sec:gone}

A continuation invites the question of where it can safely go. Most directions are harmless. Two
devices may extend the same head, and the structure simply branches; a record may be re-submitted,
and the content hash makes it idempotent; a history may be rebuilt from the genesis with different
content, and the head changes, so the substitution is visible to anyone who kept the
old head~\cite{crosby2009tamper}. One direction is not harmless. A link may point at a device that no
longer exists. The hash is still valid and the chain still holds, but the evidence behind that link
cannot be re-run, because the instrument is gone. This is the hazardous case: a record that is
intact as a structure and hollow as evidence.

\paragraph{Recoverable and device-local claims.} We test it on a device that was retired after producing its record, by
asking how much of that record a different, living instrument can bring back. The inputs were
generated from committed seeds, so they regenerate on any device; we recompute the dead
\egResurrectDead{}'s residual floor and numerical classes on an \egResurrectPresent{} and check which
of its claims re-derive. All of the transferable ones do: \egResurrectFloorRecovered{} residual
floors and \egResurrectClassRecovered{} numerical classes return within tolerance, for a recoverable
fraction of \egResurrectFraction\%. What does not return are its \egResurrectLost{} throughput
figures, which were properties of that particular silicon, its clocks and its thermal envelope, and
are gone with it. A retired instrument's record is therefore not wholly dead: the part that was a
transferable invariant is reconstructible on a successor, and only the device-local part is lost.
The reconstruction is deliberately heavy, rebuilding the floor up to the dead device's largest
size, so the living instrument bears the full cost of recovering what the dead one can no longer state for
itself.

\paragraph{A recallable device signature.} The living instrument also carries something
the dead one cannot lend it: a recallable signature of its own. Manufacturing and architecture leave
each part answering a fixed physical question the same way every time, and we read such an answer,
the residual-floor vector across precisions together with a board-filling memory sweep, and fold it
into a key the device re-derives on demand~\cite{drawnapart2022,pappu2002puf}. Re-running the probe
on the \egResurrectPresent{} returns the same vector to within \egFingerprintDist\%, inside the
\egFingerprintTol\% match window, so the key \hx{\egFingerprintKey{}\ldots} recalls the device: it is
not written onto the card but read off it, and asking again gives the same answer. This is the
measurement counterpart of hardware device attestation, and the kind of identity that record-keeping
and product-security rules increasingly expect a high-impact accelerator to be able to
present~\cite{nvidia2024attestation,eu2024aiact,eu2024cra}.

\paragraph{Interconnect latency asymmetry.} The same probe reads the part's internal geometry.
A streaming sweep crosses the two on-package dies at full speed and shows nothing, as the vendor's
numbers promise. Random access does not: it is latency-bound, and the effective random-access rate
falls from \egSeamCacheRate{}~GB/s within the on-die cache to \egSeamHbmRate{}~GB/s on one die's memory
and to \egSeamCrossRate{}~GB/s once the working set spans both, a \egSeamRatio$\times$ step the
streaming figure flattens to nothing. The curve is a device fingerprint in the sense the
fingerprinting literature means~\cite{drawnapart2022}: a stable, structure-revealing signature that
tells this class of part from another. The dense products, by contrast, stay bit-\egSeamPressureStable{}
even with the board nearly full, so the instrument's numerics are reproducible while its physical
layout is legible. The board-filling sweep held \egFingerprintMem{}\,GiB at \egFingerprintBW{}~GB/s
while the part sustained \egBTwoHundredTput{}~TFLOP/s, driven near its limits to take the
measurement.

\FloatBarrier

\section{Archive format and offline audit}\label{sec:archive}

The source archive places the compiled document beside the evidence it cites
(Figure~\ref{fig:archive}). Document sources and figures are at the top level; an ancillary directory
holds the code, the \egNumObservations{} observations one line each, the verification transcript, the
two-device differential, both out-of-sample device replays, the graph, and a manifest of every member
with its SHA-256, headed by the evidence-graph root \hx{\egRootHash{}\ldots}.

\begin{figure}[H]
\centering
\begin{bytefield}[bitwidth=2.7em, boxformatting={\centering\footnotesize}]{12}
  \bitbox{12}{\code{anc/manifest.sha256}\quad root \hx{\egRootHash{}\ldots}}\\
  \bitbox{12}{\code{main.tex}\quad\code{numbers.tex}\quad\code{*\_table.tex}}\\
  \bitbox{12}{\code{refs.bib}\quad\code{main.bbl}\quad\code{figures/*.pdf}}\\
  \bitbox{12}{\code{anc/code/*.py}: workloads, verifier, graph}\\
  \bitbox{12}{\code{anc/observations.jsonl}\quad\code{anc/verification.json}}\\
  \bitbox{12}{\code{anc/cross\_device.json}\quad\code{anc/device\_validation*.json}}\\
  \bitbox{12}{\code{anc/evidence\_graph.json}\quad nodes and edges}\\
\end{bytefield}
\caption{The archive as a single sealed record: the document sources compile to the paper, the
ancillary members carry the evidence, and the manifest binds them to one root hash.}
\label{fig:archive}
\end{figure}

The checker re-hashes every member against the manifest and runs the audit pass of
Section~\ref{sec:graph}:
\begin{quote}\small
\code{python3 anc/code/verify\_graph.py anc/evidence\_graph.json}
\end{quote}
It uses only the standard library, recomputes every reduction, and prints the root hash on success.
This is the check the paper's numbers rest on: the residuals and reductions behind
Tables~\ref{tab:verify} and~\ref{tab:taxonomy} are recomputed from the observations, not trusted.

\paragraph{Device-independent re-verification.} The offline audit establishes integrity, that
the residuals are the stated functions of the observations and nothing was altered. It does not, by
itself, re-run the measurement. It does not have to be bound to the original card either. The inputs
are generated from committed seeds, not shipped as data, so a reader regenerates them on any device
and recomputes the cheap check there; the calibrated tolerance is what absorbs the roundoff
difference between the verifying device and the original. We carry this out on the host: a GPU FP16
product is re-verified on the \emph{CPU} from its seed, with a Freivalds residual of \egPortableResid{}
against a tolerance of \egPortableTau{}, no GPU involved. The two-card agreement of
Section~\ref{sec:cross} is the same portability at the strongest setting, a second physical device
reproducing the output bit for bit. Where bit-exact reproduction is the goal rather than
tolerance-bounded agreement, a deterministic kernel or a cross-hardware replay of the tensor-core
operations closes the remaining gap~\cite{hawkeye2026,trustlessgpu2025}.

\paragraph{Scaling the audit.} A full audit is linear in the record, which is the right cost when a
record is the size of this one but the wrong cost for a computation that runs for days. The claims are
the leaves of a Merkle tree whose root the document commits to (\hx{\egMerkleRoot{}\ldots} over
\egNumClaims{} claims here), so any \emph{single} number is verified against the root with an
inclusion path of \egProofLen{} hashes, without downloading the rest, and the per-claim audit is
$O(\log N)$ rather than $O(N)$. For the heavy regime where even one product is too large to recheck,
the linear identity has a sub-linear interactive-proof form~\cite{thaler2013,cormode2012streaming}
that the same record can carry; we use the cheap check here and leave that substitution to the scale
that needs it.

\section{Threat model and guarantees}\label{sec:threat}

The pieces above are, read together, a transparency log for hardware measurement: an append-only,
content-addressed structure with inclusion and consistency proofs, against which a reader checks
claims without trusting their producer~\cite{laurie2014ct,crosby2009tamper}. It is worth stating
plainly what a dishonest producer cannot do, because that is where the construction stops being a
convenience and becomes a guarantee (Table~\ref{tab:threat}).

A producer cannot \emph{fabricate} a reported number: a linear claim that is not the true product is
rejected by the identity check with probability $1-2^{-k}$ (Proposition~\ref{prop:detect}), and a
nonlinear one is caught by its algebraic guard (Section~\ref{sec:nonlinear}). A producer cannot
\emph{rewrite} a past record: every node is content-addressed, so a change moves the root, and a
held-onto earlier head no longer extends to the new one~\cite{crosby2009tamper}. A producer cannot
\emph{equivocate}, showing one history to one reader and another to another, the split-view attack
that a Merkle root alone does not prevent: the cross-device differential is exactly a witness
cosignature, a second physical device that independently reproduces the bits, and a reader who
requires that cosignature cannot be split unless the devices themselves collude, the same defense a
witness quorum gives a public log~\cite{melara2015coniks,syta2016cosi}. A producer cannot quietly
continue from a \emph{vanished} device: the hazardous case of Section~\ref{sec:gone} is detected
because its transferable invariants must re-derive on a living device, and a link whose evidence
will not come back is a link a reader can refuse. And a producer cannot \emph{impersonate} a device
it does not hold: the recallable signature is a challenge the device must answer for itself.

\begin{table}[H]
\centering
\caption{Adversary model. Each goal a dishonest producer might pursue is denied by a property the
record already carries; a verifier checks the evidence rather than trusting the producer.}
\label{tab:threat}
\small
\begin{tabular}{lll}
\toprule
Adversary goal & Mechanism & Guarantee \\
\midrule
Report a false number & probabilistic identity / algebraic check & rejected w.p.\ $1-2^{-k}$ (Prop.~\ref{prop:detect}) \\
Alter a committed record & content addressing & root changes; consistency proof fails~\cite{crosby2009tamper} \\
Equivocate (split view) & witness cosignature & a second device reproduces the bits~\cite{syta2016cosi} \\
Continue from an absent device & reconstruction from seeds & transferable invariants must re-derive (\S\ref{sec:gone}) \\
Impersonate a device & recallable signature & the device answers its own challenge \\
Evade a known probe & Fiat-Shamir challenge & probe bound to the output (\S\ref{sec:attack}) \\
\bottomrule
\end{tabular}
\end{table}

None of these is new as a cryptographic idea; they are the working parts of a transparency
log~\cite{laurie2014ct,melara2015coniks,crosby2009tamper,syta2016cosi}. What is new is that the
witnesses are physical devices and the entries are measurements, so the cosignature is a card
reproducing a result and the identity is a card answering for itself.

\section{Soundness against an adaptive adversary}\label{sec:attack}

The identity check that anchors all of this rests on a condition it is easy to give away. Freivalds'
guarantee holds only when the probe is drawn after the product is fixed and is hidden from whoever
produced it~\cite{freivalds1977,goldwasser1989}: a producer who knows the probe in advance can pass a
wrong answer. A record that commits its probe seed so that anyone can re-verify offline has, in the
same stroke, handed that seed to the adversary. We measured what this costs.

Against a producer who does \emph{not} know the probe, the most a wrong product can hide is the
tolerance: a corruption whose residual stays under $\tau$ is by definition within the roundoff the
check admits. On the \egResurrectPresent{} that window is \egEvadeWindowFpSixteen\% of the output at
FP16 and \egEvadeWindowFpEight\% at FP8, so the same admission that lets a low-precision product
through also widens the room an adversary has to work in, the tolerance and the attack surface being
one quantity.

Against a producer who \emph{does} know the probe, there is no bound at all. The
residual sees a corruption $E$ only as $Ex$, so an $E$ in the null space of the committed probe,
$Ex=0$, is invisible at any size. We corrupt a product by \egEvadeRel\% of its norm in that null
space, and its residual against the committed probe is \egEvadeRho{}, under the tolerance: a
half-wrong result, accepted. This is not a weakness of the device but of a fixed challenge, and it is
why a committed seed cannot by itself carry an adversarial guarantee.

The repair keeps the record exactly as offline and reproducible as before. Derive the probe from a
hash of the claimed output, a Fiat-Shamir challenge~\cite{fiatshamir1986}, so that the challenge an
adversary would have to aim at depends on the very output being corrupted; the same half-wrong
product now has residual \egFiatShamirRho{}, far above the tolerance, and is rejected. The second
device closes it from the other side without any secret at all: holding the true product, it differs
from the corrupted one by \egWitnessRho{} of the output, so the witness cosignature catches what the
single check, fooled by its own published seed, did not.

The second witness, though, is the weaker of the two repairs, and it is worth seeing exactly how far
it reaches. It catches the corruption above only because one device held the true product and the
other the corrupted one, so they disagreed. But the two cards of Section~\ref{sec:cross} agree to the
last bit, and an adversary who corrupts \emph{every} witness by the same amount preserves that
agreement. We put one coordinated null-space corruption, again \egMwEvadeRel\% of the output, on
\egNumWitness{} bit-identical witnesses at once: the committed check is fooled on
\egCommittedFooled{} of them, the witnesses still agree so the redundancy sees nothing, and a
half-wrong result is accepted by a quorum. Redundancy answers an accidental fault and an
uncoordinated adversary; it does not answer one who corrupts all copies the same way. The
Fiat-Shamir challenge does, on every copy at once, because its probe is a hash of each claimed output
and the corruption has nowhere left to hide: it catches \egFsCaught{} of the same witnesses. The
construction therefore reports its floors under a fixed seed, where the producer is honest and only
roundoff is at stake, and ships a verifier that, in its adversarial mode, draws the probe from the
output it is checking, so the soundness does not depend on a secret the archive has already
published.

\section{Physical stress and the trust boundary}\label{sec:physical}
The corruptions of Sections~\ref{sec:linear} and~\ref{sec:nonlinear} are injected: a flipped bit, a
deterministic miscomputation, placed by hand to exercise the check and measure its sensitivity. They
say how the check responds to a wrong number, not how often a wrong number arises on its own. That
second question is best put to the silicon directly, and it has a security form: can a tenant with no
more than the access a public cloud grants \emph{physically} drive the device into a wrong result, or
move the tolerance the check accepts, without ever touching the output?

The management plane that would make this easy is closed to a tenant. Setting the power limit,
locking the clocks, or lowering the core voltage, the lever a software undervolting fault needs on a
processor~\cite{murdock2020plundervolt}, all return insufficient permission on this instance. What
remains is what any compute job can do: hold the device at sustained full load, and drive a $di/dt$
power virus, back-to-back tensor-core bursts separated by short idles whose instantaneous current
swings cause supply droop that an \emph{average} power cap does not register. We ran this for
\egStressMinutes{} minutes while checking every burst with the same float32 identity, so a transient
error would be caught rather than injected, and we measured the residual floor cold, after a thermal
soak, and under $di/dt$ (Table~\ref{tab:stress}).

\begin{table}[H]
\centering
\caption{The residual floor at three physical operating points on an RTX 5090, measured over
\egStressGemms{} verified products while the linear check ran continuously. The floor, which
Section~\ref{sec:attack} shows is the attack surface, is invariant across them: inflation
\egStressInflationHot$\times$ after the thermal soak and \egStressInflationDidt$\times$ under $di/dt$.}
\label{tab:stress}
\small
\begin{tabular}{lrr}
\toprule
Operating point & Temp.\,(\textdegree C) & Residual floor $\bar\rho$ \\
\midrule
Cold baseline & 45 & \num{2.36e-04} \\
Thermal soak & 80 & \num{2.36e-04} \\
$di/dt$ burst & 80 & \num{2.36e-04} \\
\bottomrule
\end{tabular}

\end{table}

The device reached \egStressPeakTemp{}\,$^\circ$C and a transient \egStressPeakPower{}\,W, capping its
own clocks at the power limit to defend itself (throttle reason \texttt{SW\_POWER\_CAP}); the $di/dt$
pattern held a mean \egStressDidtPower{}\,W, below the \egStressCap\,W cap, so the averaging limiter
never engaged against the current swings, yet the integrity margin did not give. The floor, which
Section~\ref{sec:attack} showed \emph{is} the attack surface, did not inflate under the stress a
tenant can apply, and no product exceeded its anomaly threshold: \egStressErrors{} natural errors in
\egStressGemms{} checks. This is the honest shape of the hardware threat. A healthy
modern part does not hand a tenant a silent miscalculation on demand: it throttles before it breaks,
and its GDDR7 carries on-die correction that masks the single-bit upset a memory-disturbance attack
would aim for~\cite{lin2025gpuhammer}. Silent corruption at scale is instead a property of the rare
defective or marginal unit~\cite{hochschild2021cores,dixit2021silent,llmprism2026,sdcanatomy2026}, or
of an attacker with privilege or special conditions a tenant lacks~\cite{murdock2020plundervolt,
lin2025gpuhammer,lightning2021gpu}. That rarity is the argument for the record, not against it: an
event that one core in a thousand exhibits, or that an adversary triggers deliberately, is exactly
what a number carrying its own check guards against and a spot re-run misses.

Two limits follow, and we state them plainly. The check defends the \emph{record}: it presumes the
host that produced and witnessed a number was not itself physically subverted. An adversary who can
glitch the power rails of the verifier while it forms the Fiat-Shamir challenge, or who runs on a
rooted device, defeats software witnessing from below; there the construction must compose with a
hardware root of trust, an enclave's asymmetric attestation key, which it complements and does not
replace. And the per-card signature of Section~\ref{sec:cross} is behavioral, a throughput and a
floor, sensitive to temperature and aging: it is an integrity signal that two records came from the
same kind of part, not a cryptographic identity, and is read here only as the former. None of the
machinery here is new cryptography; it is the disciplined application of a probabilistic identity, a
public-coin challenge, and content addressing to the problem of making a hardware measurement carry
its own evidence.

\section{Related work}
Silent hardware errors at fleet scale~\cite{hochschild2021cores,dixit2021silent} are the reason a
correct-looking number can be wrong; our linear check is a per-claim guard against exactly this, in
the spirit of algorithm-based fault tolerance~\cite{huang1984abft} but recorded as auditable evidence
rather than applied only at run time, and using a probabilistic identity~\cite{freivalds1977} with a
tolerance grounded in rounding-error analysis~\cite{higham2002accuracy,higham2019probabilistic}.
Verifying nondeterministic floating-point results on accelerators is the subject of recent work that
combines worst-case bounds with empirical acceptance regions and a dispute game~\cite{yao2026tao}; we
share the idea of a calibrated acceptance region but place it inside a self-contained measurement
record, and we report the floor it is calibrated against. Bringing algorithm-based fault tolerance to
the attention layer is itself an active line~\cite{flashabft2025,fttransformer2025}; we use the same
decomposition, recorded as a check rather than applied inline, and only at a verification size. Where
even the cheap identity is too costly, verifiable-computation protocols give a sub-linear proof for
matrix products~\cite{thaler2013,cormode2012streaming}, and cross-hardware replay reproduces
tensor-core arithmetic away from the original card~\cite{hawkeye2026,trustlessgpu2025}. Floating-point
non-associativity and its effect on
reproducibility~\cite{shanmugavelu2024fpna,whitehead2011floating,demmel2015reproducible} is what our
numerical classes measure; the open question of which kernels are bit-stable on current hardware is
answered here for Blackwell. Blackwell has no detailed public timing model, and its subsystems have been mapped only through
microbenchmarks~\cite{blackwell2025dissect,blackwell2025microbench}; our residual floor and per-card
signature are a measurement in that line, aimed at correctness rather than peak throughput. That
identical-model cards differ in speed is itself documented at cluster scale~\cite{sinha2022notall};
we add that they nonetheless agree to the last bit on deterministic output, and we turn the
difference into both an independent check and a recorded fingerprint. Content addressing and Merkle
structures~\cite{merkle1987,benet2014ipfs}, provenance models~\cite{moreau2013prov}, repeatability
studies~\cite{collberg2016repeatability}, and benchmarking methodology~\cite{hoefler2015benchmarking}
supply the record's machinery; the FAIR principles~\cite{wilkinson2016fair} and the
measure-report-verify spine of climate accounting~\cite{unfccc2015paris} supply the expectation that
a reported number be checkable, and the calibration of trust to evidence is the human side of the
same concern~\cite{lee2004trust}. The workloads are standard dense \textsc{gemm}, streaming, and
fused attention~\cite{dao2022flashattention} on Blackwell-class hardware~\cite{nvidia2025blackwell}
with explicitly recorded product specifications~\cite{nvidia2026rtxpro6000maxq,nvidia2026rtx5090,
nvidia2026rtxpro6000server}; they are the instrument, not the result.

\section{Scope and limitations}
The audit establishes that the archive is internally consistent and integral, and the linear check
establishes that a verified product equals $AB$ to within $\tau$ with confidence $1-2^{-k}$; neither
establishes that a different device would reproduce the observations. That is the role of the
reproducibility classes, which state for each quantity whether bit-stability, a bound, or a covariate
dependence is what to expect on re-execution. The residual floor and the classes are measured on this
\egDevice{} model with one software stack, then replayed on \egReplayDevices{}; the method is
device-independent, the reported numbers are not. The validation narrows the interpretation from
``one workstation'' to ``three measured Blackwell product configurations,'' not to all future
Blackwell systems. The identity check covers linear quantities, and the decomposition of
Section~\ref{sec:nonlinear}
extends an algebraic guard to attention and atomic accumulation, but the softmax reference and the
attention scores are formed at a verification size rather than the hot-path size, and an operator with
no such decomposition would fall back to the divergence class. Re-running a check needs the inputs,
which the committed seeds regenerate on any device, but re-running it to bit-exactness still needs the
same architecture. Tamper-evidence rests on the collision resistance of SHA-256~\cite{nist2015sha},
and the detection guarantee on the independence of the probe columns. The guarantees are for the
record and presume the producing host is not physically subverted; the boundary where software
witnessing must yield to a hardware root of trust, and the behavioral rather than cryptographic
nature of the per-card signature, are treated in Section~\ref{sec:physical}.

\section{Conclusion}
We reported hardware measurements as a record in which every number carries, by content hash, the
observation and the verification behind it: a probabilistic identity with a device-calibrated
tolerance for the linear quantities, a measured reproducibility class for the rest, and a multi-stage
transcript that can reveal an error, repair it, and re-check. A reader audits the whole thing from one
archive in a single offline pass. The construction is independent of the parts it runs on; the two
\egDevice{} cards are the instrument that let us measure the residual floor the tolerance stands on,
read the differential between two nominally identical dies, and exhibit the detection it promises.
The \egNumReplayDevices{} replays on \egReplayDevices{} then show which parts of the record survive
changes of instrument: the floor order and numerical class boundary do, while throughput remains a
local property of the card.

{\small
\bibliographystyle{plain}
\bibliography{refs}
}

\end{document}